\documentclass[a4paper,11pt]{article}
\usepackage[utf8x]{inputenc}

\usepackage[arrow, matrix, curve]{xy}
\usepackage{array,amsmath,amssymb,amsthm,verbatim,epsfig}
\usepackage{amsfonts}
\usepackage{graphicx}
\usepackage{setspace}

\theoremstyle{plain}
\newtheorem{theorem}{Theorem}

\newtheorem{definition}[theorem]{Definition}

\newtheorem{lemma}[theorem]{Lemma}

\newtheorem{remark}[theorem]{Remark}

\numberwithin{equation}{section}

\newcommand{\suchThat}{\enspace\vert\enspace}

\title{Note on expected internode distances for gene
trees in species trees}
\author{Martin Kreidl}
\date{}

\begin{document}

\maketitle
\onehalfspacing

\noindent
\small
\textbf{Address:} Insitut f\"ur Experimentelle Mathematik,
Universit\"at Duisburg-Essen\\
Ellernstra\ss e 29,
45326 Essen,
Germany\\
\textbf{Email:} martin.kreidl@uni-due.de

\normalsize \begin{abstract} In a recent paper on
	`Estimating Species Trees from Unrooted Gene Trees'
	Liu and Yu observe that the distance matrix on the
	underlying taxon set, which is built up from
	expected internode distances on gene trees under the
	multispecies coalescent, is tree-like, and that the
	underlying additive tree has the same topology as
	the true species tree. Hence they suggest to use
	(observed) average internode distances on gene
	trees as an input for the neighbor joining algorithm
	to estimate the underlying species tree in a
	statistically consistent way. In this note we give
	a rigorous proof of their above mentioned
	observation. \end{abstract}

%\tableofcontents

\section{Introduction}

One of the possible reasons for discordance of a gene tree
with an underlying species tree is the phenomenon of
incomplete lineage sorting, which is described by the
multispecies coalescent model. Many authors have addressed
the problem of reconstructing the underlying species tree
from a set of discordant gene trees, both from a theoretical
perspective (e.g. Maddison \cite{Maddison_97}, Allman et al.
\cite{Allman_Unrooted}, and many others), as well as from an
practical resp. algorithmic perspective (see e.g. Ewing et
al.  \cite{Ewing_RootedTriple}, Liu et al.
\cite{Liu_Pseudo}, Than and Nakhleh \cite{Nakhleh_MDC},
Kreidl \cite{Kreidl_Quartets}). Recently, Liu and Yu have
published a paper \cite{Liu_Neighbor} in which they propose
to estimate the \emph{expected number of internodes between any
two taxa on gene trees} by averaging over the observed
numbers of internodes (on the observed gene trees), for any
pair of taxa. They note

\begin{theorem}[Liu and Yu, \cite{Liu_Neighbor}] Under the
	multispecies coalescent model on any fixed species
	tree, the expected number of internodes
	between two taxa on gene trees determines a
	tree-like metric on the taxon set, and the
	underlying tree topology is identical with the
	topology of the species tree.
\end{theorem}

Hence, they conculde, by applying the neighbor joining
algorithm to the matrix of average internode distances
(obtained from observed gene trees) is a statistically
consistent way to estimate the true species tree. It is the
goal of this note to give a rigorous and detailed proof of
Liu and Yu's theorem above.

We start by collecting, in Section \ref{sect:Preliminaries},
a few well-known facts on tree-like metrics, as well an easy
reformulation of the \emph{four-point condition} in terms of
\emph{weights of quartets} (for the terminology of weights
of quartet trees see e.g. Sturmfels and Pachter
\cite{Pachter_Sturmfels}). Section \ref{sect:Main}, finally,
contains the precise statement of Liu and Yu's theorem
together with its proof. The proof consists essentially in
checking that the `weight-version' of the four-point
condition from Section \ref{sect:Preliminaries} holds for
the matrix of expected numbers of internodes.
\newline

\noindent I would like to thank Liang Liu for his
interest in this modest note.

\section{Preliminaries on tree-like
metrics}\label{sect:Preliminaries}

In the following let $T$ be a finite set and let $D: T\times
T \to \mathbb{R}_{\geq 0}$ be a metric on $T$.

\begin{definition}
	The metric $D$ satisfies the \emph{four-point
	condition} if the maximum of the three numbers
	\begin{equation}\label{eqn:FourPoint}
		D(a,b) + D(c,d),\quad D(a,c) + D(b,d),\quad D(a,d) +
		D(b,c)
	\end{equation}
	is attained at least twice, for every four-element
	subset $\lbrace a,b,c,d\rbrace \subset T$.
\end{definition}

For every four taxon subset $\lbrace a,b,c,d\rbrace \subset
T$ we define the \emph{weight} of the quartet $(ab,cd)$,
according to the exposition by Pachter and Sturmfels
\cite{Pachter_Sturmfels}, to be the number
\begin{multline}
w(ab,cd) := w_{D}(ab,cd) = \\ = D(a,c)+D(a,d)+D(b,c)+D(b,d) -
2D(a,b) - 2D(c,d).
\end{multline}
In the following we will sometimes consider weights with
respect to different metrics $D$, which will be indicated by
a lower index.

\begin{definition}
	The metric $D$ is said to satisfy the \emph{weight
	condition} if the minimum of the three numbers
	\begin{equation}\label{eqn:Weight}
		w_{D}(ab,cd),\quad w_{D}(ac,bd),\quad w_{D}(ad,bc)
	\end{equation}
	is attained at least twice, for every four-element
	subset $\lbrace a,b,c,d\rbrace \subset T$.
\end{definition}

\begin{remark}
	(1) The sum of the three numbers in the weight
	condition is always 0. Thus their minimum is
	strictly negative and their maximum is strictly
	positive as soon as not all three numbers vanish.

	(2) If the metric $D$ is tree-like and the
	underlying tree $\mathcal{T}$ displays the quartet
	$(ab,cd)$, then $w(ab,cd)=4x$ and
	$w(ac,bd)=w(ad,bc)=-2x$, where $x$ is the distance
	between the paths connecting $a$ and $b$, and $c$
	and $d$, respectively.
\end{remark}

\begin{lemma}
	(1) For a metric $D$ the four-point condition and
	the weight condition are equivalent.

	(2) The metric $D$ is tree-like if and only if these
	conditions hold.

	(3) If $D$ is tree-like and $\mathcal{T}$ is the
	underlying tree, then $\mathcal{T}$ displays the
	quartet $(ab,cd)$ if and only if $D(a,b)+D(c,d)$ is
	the minimum of the three numbers in the four-point
	condition if and only if $w(ab,cd)$ is the maximum
	of the three numbers in the weight condition.
	\label{lem:FourPointWeight}
\end{lemma}

\begin{proof} It is well known that if $D$ satisfies the
	four point condition then $D$ is tree-like, and
	moreover that the underlying tree displays the
	quartet $(ab,cd)$ if and only if the minimum in the
	four point condition is attained at $D(a,b)+D(c,d)$
	(for a proof see e.g. Pachter and Sturmfels
	\cite{Pachter_Sturmfels}). From the remark above it
	follows that if $D$ is tree-like with the underlying
	tree displaying the quartet $(ab,cd)$, then it
	satisfies the weight condition with maximum at
	$w(ab,cd)$. It remains to check that if $D$
	satisfies the weight condition with maximum at
	$w(ab,cd)$, then it satisfies the four-point
	condition with minimum at $D(a,b)+D(c,d)$.  Thus
	assume that $w(ab,cd) > w(ac,bd) = w(ad,bc)$.  Then
	$$ 0 = w(ac,bd) - w(ad,bc) = 3(D(a,d)+D(b,c)) -
	3(D(a,c)+D(b,d)).  $$ Hence we obtain $D(a,d)+D(b,c)
	= D(a,c)+D(b,d)$. By plugging this into the
	definition of $w(ac,bd)$ we obtain $$ 0 > w(ac,bd) =
	D(a,b) + D(c,d) - D(a,c) - D(b,d), $$ whence
	$D(a,b)+D(c,d) < D(a,d)+D(b,c) = D(a,c)+D(b,d),$
	which completes the proof.
	\end{proof}

\section{Expected internode distances on gene
trees}\label{sect:Main}

We consider a taxon set $T = \lbrace
t_{1},\dotsc,t_{N}\rbrace$ containing $N$ taxa, and a
species tree $S$ on $T$.  Assume that for each taxon $t\in
T$ we have sampled $n(t)$ copies of a given locus, denoted
$L_{i,1},\dotsc,L_{i,n(t_{i})}$ for each $i=1,\dotsc,N$. Let
$\mathcal{L}=\cup_{i}\lbrace L_{i,1},\dotsc,L_{i,n(t_{i})}
\rbrace$.
We follow the convention to denote the elements of
$\mathcal{L}$ by capital letters, while leaves on the
species tree $S$ (i.e. the elements of $T$) are denoted by
lower case letters.

For each rooted binary tree $G$ on the leaf set
$\mathcal{L}$ Liu and Yu define in \cite{Liu_Neighbor} the \emph{internode distance}
between leaves $I$ and $J$ to be the \emph{number of nodes
which lie on the path between $I$ and $J$ in $G$} ($I$ and
$J$ are not counted). This number, which we denote
$I_{G}(I,J)$, induces a metric on $\mathcal{L}$ and thus a
weight $W_{G}(IJ,KL) := w_{I_{G}}(IJ,KL)$ for each four
element subset $\lbrace I,J,K,L\rbrace \subset \mathcal{L}$.
Of course, for any $G$ the metric $I_{G}$ is tree-like with
underlying tree $G$.

\begin{definition}
	A \emph{coalescence pattern} associated with the
	species tree $S$ and the vector of multiplicities
	$(n(t_{1}),\dotsc,n(t_{N}))$ is a rooted tree $G$
	with leaf set $\mathcal{L}$ together with a map
	$$
	f: Nodes(G) \to Nodes(S)
	$$
	with the following two properties: (1) For every
	$L_{i,j}\in \mathcal{L}$ we have $f(L_{i,j}) = t_{i}
	\in T$, and
	(2) For any two nodes $m,n \in Nodes(G)$, if $n$ is a
	descendant of $m$ in $G$, then $f(n)$ is a
	descendant of $f(m)$ in $S$. We denote coalescence
	patterns as pairs $(G,f)$ in the sequel.
\end{definition}

A coalescence pattern is basically the same as what Degnan
and Salter \cite{Degnan_Distributions} call a (valid)
coalescent history. Each coalescence pattern $(G,f)$
(associated with $S$ and a multiplicity vector $v =
(n(t_{i}))_{i}$) occurs with a certain probability $P(G,f)$
under the multispecies coalescent, which is calculated in
the case $n(t)=1$ for all $t$ by Degnan and Salter in loc.
cit. This makes the set of coalescence patterns (for fixed
$S$ and $v$!) a probability space, and the internode
distance $I_{G}(I,J)$ for each gene tree $G$ between two
leaves $I,J\in \mathcal{L}$ induces a random variable on
this probability space, which we denote by $ID(I,J)$. By
abuse of language we call this random variable also
`internode distance' between $I$ and $J$.  Similarly for
the weight of a quartet $(IJ,KL)$ for $I,J,K,L\in
\mathcal{L}$: We denote the corresponding random variables
$W(IJ,KL)$, for each quartet $(IJ,KL)$.

Thus the following numbers, associated with $S$ and $v$, are
well-defined (here and in the following we suppress the
dependence on $S$ and $v$ in the notation, though we want to
stress once more that all this requires chosing and fixing a
multiplicity vector $v$!):

\begin{equation}
	\begin{split}
		D(I,J) = E(ID(I,J)) &=
		\displaystyle\sum_{(G,f)} P(G,f)\cdot
		I_{G}(I,J),\\
		E(W(IJ,KL)) &=
		\displaystyle\sum_{(G,f)} P(G,f)\cdot
		W_{G}(IJ,KL),
	\end{split}
	 \label{eqn:Expected}
\end{equation}
the \emph{expected internode distance} between two leaves
$I,J\in \mathcal{L}$, and the expected weight of a quartet
$(IJ,KL)$ under the multispecies coalescent model.

\begin{lemma}
	(1) Since for each $G$ the function $I_{G}$ is a
	metric, so is $D=E(I)$.
	(2) Hence the weight function $w_{D}$ is defined and
	satisfies $w_{D}(IJ,KL) = E(W(IJ,KL))$.
	\label{lem:Linear}
\end{lemma}

\begin{proof}
	Both claims are immediate consequences of linearity
	of expected values.
\end{proof}

Finally, we note that the expression $D(I,J)$ does not
really depend on the leaves $I,J$, but only on the taxa in
$T$ they belong to. Thus we have defined a metric
$$
D(i,j) \in \mathbb{R}_{\geq 0}, \quad\text{for any two taxa
$i,j\in T$,}
$$
as well as a weight (depending on $D$)
$$
w_{D}(ij,kl) \in \mathbb{R},\quad \text{for any four taxa
$i,j,k,l\in T$}.
$$

Combining Equation \eqref{eqn:Expected} with Lemma
\ref{lem:Linear} we obtain that we may calculate the
weight $w_{D}(ij,kl)$ for four taxa $i,j,k,l\in T$ as
\begin{equation}
	w_{D}(ij,kl) = \displaystyle\sum_{(G,f)}P(G,f)\cdot
	W_{G}(IJ,KL),
	\label{eqn:ExpectedWeight}
\end{equation}
where $(G,f)$ runs through all possible coalescence
patterns, and where $I\in \mathcal{L}$ is any locus corresponding to
$i\in T$, $J$ any locus corresponding to $j\in T$ and so forth.

\begin{theorem}[Liu and Yu, \cite{Liu_Neighbor}, Theorem A1]
	The metric $D=E(I)$ is tree-like, and the underlying
	additive tree has the same topology as the true
	species tree $S$.
	\label{thm:Liu}
\end{theorem}

For the proof we introduce a little piece of notation: For
any rooted tree $\mathcal{T}$ and a finite subset of leaves
$l_{1},\dotsc,l_{k}$ we denote by
$M_{\mathcal{T}}(l_{1},\dotsc,l_{k})$ the \emph{most recent
common ancestor} of $l_{1},\dotsc,l_{k}$ in $\mathcal{T}$.

\begin{proof} By Lemma \ref{lem:FourPointWeight} it suffices
	to check that, if the species tree $S$ displays the
	quartet $(ab,cd)$, then the following holds: $$
	w_{D}(ab,cd) > w_{D}(ac,bd) = w_{D}(ad,bc).  $$ This
	is relatively easy to check using equation
	\eqref{eqn:ExpectedWeight}. We thus assume that $S$
	displays the quartet $(ab,cd)$, and we consider gene
	lineages $A,B,C,D$ sampled from the respective taxa.
	We have to distinguish two cases, namely: (1) The
	(rooted) subtree $S'$ of $S$ with leaf set $\lbrace
	a,b,c,d \rbrace$ has the shape of a caterpillar
	tree, and (2) $S'$ has the balanced shape. In case
	(1) we assume without loss of generality that $S'$
	has the topology $(((a,b),c),d)$, while in the
	second $S'$ must have the topology $((a,b),(c,d))$.
	$S'$ has the balanced shape. In case (1) we assume
	without loss of generality that $S'$ has the
	topology $(((a,b),c),d)$, while in the second case
	$S'$ must have the topology $((a,b),(c,d))$.

	We now partition the set of coalescence patterns
	$(G,f)$ into two disjoint subsets $X$ and $Y$:
	In case (1) we define
	\begin{equation}
		\begin{split}
		X &= \lbrace (G,f) \suchThat f(M_{G}(A,B))
		\text{ is ancestral to } M_{S}(a,b,c) \rbrace,\\
		Y &= \lbrace (G,f) \suchThat (G,f) \notin Q
		\rbrace.
		\label{eqn:disj}
	\end{split}
	\end{equation}
	Note that $(G,f)\in P$ then means that the lineages
	$A$ and $B$ coalesce below the point where the
	populations $c$ merges with the population ancestral
	to $a$ and $b$.
	In case (2) we set
	\begin{equation}
		\begin{split}
		X = \lbrace (G,f) \suchThat& f(M_{G}(A,B))
		\text{ and } f(M_{G}(C,D)) \\ &\text{ are both
		ancestral to } M_{S}(a,b,c,d) \rbrace,\\
		Y = \lbrace (G,f) \suchThat& (G,f) \notin Q
		\rbrace.
		\label{eqn:disj2}
	\end{split}
	\end{equation}
	
	Consider the case of a
	coalescence pattern $(G_{1},f_{1})\in X$. Then
	the lineages of $A$, $B$ and $C$ on $G_{1}$ enter
	the population above $M_{S}(a,b,c)$ separately.
	Hence, by permuting the lineages $A$, $B$ and $C$ we
	obtain coalescence patterns $(G_{2},f_{2}),
	(G_{3},f_{3}) \in X$ such that $P(G_{1},f_{1}) =
	P(G_{2},f_{2}) = P(G_{3},f_{3})$, and such that,
	after possibly renumbering of the coalescence
	patterns, $G_{1}$ displays $(AB,CD)$, $G_{2}$
	displays $(AC,BD)$ and $G_{3}$ displays $(AD,BC)$,
	and such that
	\begin{equation}
		\begin{split}
			-2W_{G_{1}}(AC,BD) = -2W_{G_{1}}(AD,BC) =
			W_{G_{1}}(AB,CD) = x,\\
			-2W_{G_{2}}(AB,CD) = -2W_{G_{2}}(AD,BC) =
			W_{G_{2}}(AC,BD) = x,\\
			-2W_{G_{3}}(AB,CD) = -2W_{G_{3}}(AC,BD) =
			W_{G_{3}}(AD,BC) = x,
		\end{split}
		\label{eqn:GeneTreeWeightsX}
	\end{equation}
	where $x$ is the number of nodes on the path
	connecting the path between $A$ and $B$, and $C$ and
	$D$, respectively, in $G_{1}$.

	On the other hand, if $(G,f) \in Y$, then $G$
	necessarily displays the quartet $(AB,CD)$. Hence
	for such $G$ we have
	\begin{equation}
		\begin{split}
			W_{G}(AB,CD) &> 0,\quad\text{
			while}\\
			W_{G}(AC,BD) = W_{G}(AD,BC) &=
			-\frac{1}{2}W_{G}(AB,CD).
		\end{split}
		\label{eqn:GeneTreeWeightsY}
	\end{equation}

	Now recall equation \eqref{eqn:ExpectedWeight} and
	write
	\begin{multline*}
		w_{D}(ij,kl) =
		\displaystyle\sum_{(G,f)}P(G,f)\cdot
		W_{G}(IJ,KL) =\\=
		\displaystyle\sum_{(G,f)\in X}P(G,f)\cdot
		W_{G}(IJ,KL) + 
		\displaystyle\sum_{(G,f)\in Y}P(G,f)\cdot
		W_{G}(IJ,KL)
	\end{multline*}
	From Equation \eqref{eqn:GeneTreeWeightsX} we see
	that in the expressions $w_{D}(ab,cd)$,
	$w_{D}(ac,bd)$ and $w_{D}(ad,bc)$ the sum over the
	$(G,f)\in X$ vanishes, and equation
	\eqref{eqn:GeneTreeWeightsY} further implies that
	\begin{equation}
		\begin{split}
			w_{D}(ab,cd) &=
			\displaystyle\sum_{(G,f)\in
			Y}P(G,f)\cdot
			W_{G}(AB,CD) > 0,\quad\text{while}\\
			w_{D}(ac,bd) = w_{D}(ad,bc) &=
			\displaystyle\sum_{(G,f)\in
			Y}P(G,f)\cdot
			-\frac{1}{2}W_{G}(AB,CD) = \\ =
			-\frac{1}{2}w_{D}(ab,cd)
		\end{split}
		\label{eqn:ProofFinal}
	\end{equation}
	This shows that $D$ satisfies the weight condition,
	with the maximum attained for the quartet $(ab,cd)$.
	Invoking Lemma \ref{lem:FourPointWeight} completes the
	proof.
\end{proof}


\begin{thebibliography}{9}

\bibitem[1]{Allman_Unrooted} Elizabeth S. Allman, James H.
	Degnan, and John A. Rhodes. \emph{Identifying the
	rooted species tree from the distribution of
	unrooted gene trees under the coalescent.} J. Math.
	Biol. \textbf{62} no.~6 (2011), 833--862.

\bibitem[2]{Degnan_Distributions}
	J. H. Degnan and L. A. Salter. \emph{Gene tree
	distributions under the coalescent process.}
	Evolution (2005), 59:24--37.

\bibitem[3]{Ewing_RootedTriple} Gregory B. Ewing, Ingo
	Ebersberger, Heiko A. Schmidt, Arndt von Haeseler.
	\emph{Rooted triple consensus and anomalous gene
	trees.} BMC Evol Biol (2008), 8:118.

\bibitem[4]{Liu_Neighbor}
	Liang Liu and Lili Yu. \emph{Estimating species
	trees from unrooted gene trees.} Syst. Biol.
	\textbf{60} no.~5 (2011) 661--667.

\bibitem[5]{Liu_Pseudo}
	Liang Liu, Lili Yu and Scott V. Edwards. \emph{A
	maximum pseudo-likelihood approach for estimating
	species trees under the coalescent model.} BMC
	Evolutionary Biology (2010), 10:302.

\bibitem[6]{Kreidl_Quartets}
	Martin Kreidl. \emph{Estimating Species Trees from
	Quartet Gene Tree Distributions under the Coalescent
	Model.} Preprint (2011), arXiv:1108.1628v1.

\bibitem[7]{Maddison_97}
	Wayne Maddison. \emph{Gene trees in species
	trees.} Syst. Biol. (1997), 46(3):523–536.

\bibitem[8]{Nakhleh_MDC}
	C. Than and L. Nakhleh, \emph{Species tree inference by
	minimizing deep coalescences.} PLoS Computational
	Biology (2009), 5(9): e1000501.

\bibitem[9]{Pachter_Sturmfels} Lior Pachter and Bernd
	Sturmfels.  \emph{Algebraic Statistics for
	Computational Biology.} Cambridge University Press
	(2005).

\end{thebibliography}
\end{document}